\documentclass[10pt,conference,letterpaper]{IEEEtran}
\usepackage{times,amsmath,epsfig,amsthm,braket, amssymb, courier}
\usepackage{algorithm}
\usepackage{algpseudocode}
\usepackage{cite}
\title{Secure Logical Schema and Decomposition Algorithm for Proactive Context Dependent Attribute Based Access Control }
\author{%
{U\u{g}ur Turan{\small $~^{\#1}$}, \.{I}smail Hakk{\i} Toroslu{\small $~^{\#2}$} }%
\vspace{1.6mm}\\
\fontsize{10}{10}\selectfont\itshape
$^{\#}$\,Department of Computer Engineering, Middle East Technical University\\
06800, Ankara, Turkey\\
\fontsize{9}{9}\selectfont\ttfamily\upshape
$^{1}$\,ugur.turan@ceng.metu.edu.tr\\
$^{2}$\,toroslu@ceng.metu.edu.tr%

}
\theoremstyle{definition}%
\newtheorem{defn}{Definition}
\newtheorem{thr}{Theorem}

\begin{document}
\maketitle
\begin{abstract} 
Traditional database access control mechanisms use role based methods, with generally row based and attribute based constraints for granularity, and privacy is achieved mainly by using views. However if only a set of views according to policy are made accessible to users, then this set should be checked against the policy for the whole probable query history. The aim of this work is to define a proactive  decomposition algorithm according to the attribute based policy rules and build a secure logical schema in which relations are decomposed into several ones in order to inhibit joins or inferences that may violate predefined privacy constraints. The attributes whose association should not be inferred, are defined as having security dependency among them and they form a new kind of context dependent attribute based policy rule named as security dependent set. The decomposition algorithm works on a logical schema with given security dependent sets and aims to prohibit the inference of the association among the elements of these sets. It is also proven that the decomposition technique generates a secure logical schema that is in compliance with the given security dependent set constraints. 
\end{abstract}

%
\section{Introduction}
Business technology era has increased the importance of logical data storage and retrieval from the point of security, since many users and roles with different access privileges act in the same database environment. As an important topic in security, granularity is also essential in database access control methods. Traditional database security approaches mainly use relation based action rules for users such as allowing querying but disallowing updating the relation and sometimes they also try to define policy rules on attributes to increase granularity\cite{r1}. These approaches have very simple purpose; that is, to determine whether to grant or deny the access, based on the predefined constraints related to the role of the user.  Especially for the attribute based access control; the attributes that are going to be related with each other by executing a query, also taking query history into consideration, are the main factor in making the grant or deny decision of the query. 

For example, let$'$s consider a relation:\newline

\texttt{STUDENT = (\underline{\emph{id}}, \emph{email}, name, surname,} 

\hspace{2.2cm} \texttt{address, age, gender)}\newline

\noindent
on which a  survey about the characteristics of students are being carried out. As an example, the extraction of \texttt{email} and \texttt{gender} relationship should be forbidden in order to preserve the assumed privacy of a student. In addition to that, if \texttt{id} and \texttt{email} fields are two keys and \texttt{id} is selected as primary key, probable join queries should also be checked in order to guarantee that \texttt{email} and \texttt{gender} fields cannot be related with each other by the help of the \texttt{id} attribute. For this case, decomposing the \texttt{STUDENT} relation into two views as:\newline

\texttt{STUDENT$_1$ = (\underline{\emph{id}}, name, surname, address, }

\hspace{2.4cm} \texttt{age, gender)}

\vspace{1.6mm}

\texttt{STUDENT$_2$ = (\underline{\emph{id}}, \emph{email}, name, surname,  }

\hspace{2.4cm} \texttt{address, age)} \newline

is an example of faulty decomposition since \texttt{email} and \texttt{gender} can be related as follows. \newline

\hspace{1cm} \texttt{SELECT s1.gender, s2.email}

\hspace{1cm} \texttt{FROM STUDENT$_1$ s1, STUDENT$_2$ s2}

\hspace{1cm} \texttt{WHERE s1.id = s2.id} \newline

To prevent this kind of queries, a correct decomposition can be given as:\newline

\texttt{STUDENT$_1$ = (name, surname, address, }

\hspace{2.4cm} \texttt{age, gender)}

\vspace{1.6mm}

\texttt{STUDENT$_2$ = (\underline{\emph{email}}, name, surname, }

\hspace{2.4cm} \texttt{address, age)}

\vspace{1.6mm}

\texttt{STUDENT$_3$ = (\underline{\emph{id}}, name, surname, address, }

\hspace{2.4cm} \texttt{age)} \newline

Note that \texttt{STUDENT$_1$} is a keyless relation as it is usual in views and the situation will be discussed in the following sections of the paper. In addition to that, it is assumed there exists no functional dependency other than the ones which make \texttt{email} and \texttt{id} candidate keys for the relation. By this decomposition, the queries which relates \texttt{email} and \texttt{gender} cannot achieve the inference of association among the attributes, since equijoins on keys cannot be done by decomposed relations.

As in this example, a view based solution can be generated to satisfy privacy policy, which is very popular approach in enterprise database systems. There can be several policy rules, and views should be constructed in order to satisfy all constraints of these policies. The need for defining different external layers for different access control policies has increased by web based data sharing trend\cite{r24}. Therefore, a formal approach is needed to build a secure external layer by decomposing the relations into sub relations according to policy rules, in order to generate relevant secure logical schema. 

Most of the research on this kind of security (access control) is mainly focused on dynamic mechanisms employing query investigation or modification methods, and by also tracking the query history \cite{r1, r3, r5, r6, r9, r10}. On the other hand, the strategy used in this paper is to decompose the relation into views in advance, for preventing the time spent by query modification or history tracking operations which may be costly  in high utilized database systems \cite{r22}. To the best of our knowledge, this is the first attempt in the literature to handle privacy for context dependent attribute based access control by a proactive approach. Our method can be easily adapted for validating existing external schema against given attribute based policy rules. The proactive decomposition method described in this paper can easily be combined with other constraints such as row based policy rules during implementation. In addition to that the method can be used in a "Private Record Matching" engine when required context dependent attribute sets are supplied \cite{r13}. 

In the rest of this paper, we first present a formal method to define secure logical schema for preserving context dependent attribute based privacy, then we define a decomposition algorithm that guarantees to produce a secure logical schema. A detailed real life example is also given to clearly show the steps of the algorithm and use of the decomposed relations in sample applications. 

This paper is organized as follows: Section 2 describes the related works in the field of security and privacy in databases. Then, the Section 3 gives the preliminaries and definitions used through the rest of the paper. Section 4 presents the decomposition algorithm that satisfies the access policies, and the proof of the algorithm. The next section, Section 5, contains a real life example for demonstration of the algorithm and Section 6 briefly discusses the future work. Finally Section 7 contains the conclusion.

\section{Related Work}

The field of database security is very popular, and several works in this field have influenced the idea proposed in this paper \cite{r1,r2,r3,r4,r5,r6,r7,r8,r9,r10,r11,r14,r15,r16,r17,r18,r19,r20}. The approach of updating the query dynamically depending on the context and the policy has been studied for a long time in the literature \cite{r5}.  In this method, the query can be modified by adding predicates and the main purpose is the row based security. Adding more predicates to \texttt{where} clause can only restrict the rows extracted by the query \cite{r5}. Actually the security mechanism in \cite{r5} gives user a set of views which are permitted to be queried and then performs row based elimination by adding predicates whose idea can be treated as an additional functionality for the work in this paper. However another work, \cite{r6}, states that the former algorithm is not maximal and limits some permitted answers. In \cite{r6} some flexibility has been added as the query may depend on any view or sub view or meta-relations. That means extra work should be conducted in order to find which permitted views are involved. These two approaches may have performance problems and modifying the query can be costly \cite{r22}, nevertheless it should be noted that their query modification strategies are done mainly for row based access control, whereas this paper focuses on context dependent attribute based access control in a proactive manner. 

In addition to this, Oracle presents Virtual Private Database term \cite{r10} and performs the security totally by query modification on real relations. The modification can be as row based by adding predicates or column based by making null of the unwanted attributes. Bertino \cite{r8} calls this type of query modification approaches as “Truman Models” \cite{r7}, since they answer each time, nevertheless the answer may not be maximal because of restrictions.  These models have simple attribute based policy rules as just checking the existence of attributes in the query result. Beside this, data perturbation \cite{r16} is another run-time consuming method and may be used for Truman Models. In addition to that, “k-anonymity” term \cite{r3} has been proposed to divide the relation to views which are targeted not to extract "\textit{id}"s. Moreover the security policy need not to be satisfying the anonymity only; for instance one can define a policy rule as \textit{gender} and \textit{address} should not be obtained together as even both of them is not adequate for identification \cite{r15}. 

Furthermore, Purpose Oriented Access Control scheme \cite{r2} offers role - purpose - column mapping, however two purposes may serve to another unwanted purpose. For instance let \texttt{a, b, c} to be attributes and purpose-1 needs \texttt{a, b}; purpose-2 needs \texttt{a, c} and non-existing and unwanted purpose-3 needs \texttt{b, c}. In this example first two purposes can serve to the unwanted third purpose. That example presents the notion of query history \cite{r17} whose deep investigation makes the computation costly. To get rid of these, attribute mutability term \cite{r4} provides a mechanism as Chinese-Wall method \cite{r11} with historical data, but performance requirements may be again critical. 

Beside this, “Non-Truman” models have been proposed \cite{r7} which reject the unauthorized query according to the authorized views. Hippocratic databases \cite{r9} combine many security issues stated in this section however the addressed problem in this paper is a bit different. The main problem is to maintain security and privacy in all these works; nevertheless, dynamic security modeling with query modification, attribute mutation, historical query tracking or grant/reject mechanism may have performance problems because of their run time executions. This paper constructs a proactive security mechanism as building an external layer with a secure logical schema to user by a decomposition algorithm in which user is free to query anything on decomposed relations. The term “Attribute-Based” in this paper is used for the ability of defining the access control rules on the attributes of relations. The same term has been used differently in \cite{r20} to build the access control with the help of dedicated attributes. It is important to note that the notion of modeling  access control rules on attributes according to the application semantics is another important work discussed in \cite{r21} which is not in the scope of this paper.

The most relevant study, targeted a similar problem with this work,  is reducing inference control to access control \cite{r1}. However their solution labels the normalized schema relations and the solution is not proactive, only more efficient than query controlling.

\section{Preliminaries and Problem Definition}

In this section, we give the basic terms and concepts used in the paper. This paper has two main objectives, namely, formally defining a secure logical schema which is in compliance with the given security constraints (security dependent sets), and developing a decomposition algorithm which divides relations into sub-relations to be able to satisfy the security constraints. The main reason for decomposition is to prevent obtaining securely dependent attributes together directly in a relation or through a join.

Therefore, first, the definition of the logical schema is given in terms of two sets as relational schema and (non-reflexive and non-partial) functional dependencies. After that, the closures of relations and functional dependencies (again non-reflexive and non-partial)  are defined. The closure of relation schema is very important, since it describes how new relations can be generated using only equijoins on foreign keys. Moreover, the closure of functional dependencies is used to define identifiers for attributes, how they can be inferred, and how two or more attributes can be associated with each other. Combining these definitions, we then define a secure logical schema, which simply prevents obtaining the attributes of each given security dependent set together by joins. We also prove that secure logical schema guarantees that it is not possible to obtain any association among the set of attributes of security dependent set. 

Following these, we define a decomposition operation which decomposes a logical schema according to a given secure dependent sets in order to form a secure logical schema. Afterwards, we prove that the new schema obtained by employing the decomposition operation is secure logical schema, which means that it is not possible to associate attributes of secure dependent sets by joins using the relations constructed after the decomposition. By this way, the inference of association of the attributes together in each security dependent set can be prevented.

\vspace{1.6mm}

\begin{defn}[\textbf{Relational Schema}]
A \textit{relation schema} is defined a set of attribute names concatenated with relation name (using underscore) in order to prevent the vagueness caused by having same attribute name in different relations. For the sake of simplicity, \textit{relation schema} is referred as \textit{relation} and the concatenation on attribute names will not be shown unless needed.
\end{defn}
For example a relation
\begin{multline*}
\texttt{USERS} \texttt{=} \texttt{\{\underline{\emph{id\_users}},} \texttt{name\_users,} \texttt{surname\_users,} \\ \texttt{email\_users\}}
\end{multline*}
is defined as a set of concatenated attribute names. Using this definition, it is guaranteed that all attribute names in a database will be unique owing to unique relation names by default.

\begin{defn}[\textbf{Logical Schema}]
A \textit{logical schema} for a database is defined as a tuple $\mathcal{L=(R,F)}$ such as;

\begin{itemize}
\item	$\mathcal{R}$ is defined as set of all relation schemas in a database. 
\item	$\mathcal{F}$ is defined as set of functional dependencies among attributes in all relation schemas in $R$ excluding reflexive and partial functional dependencies.
\end{itemize}

\end{defn}

Since a foreign and relevant key in two different relations have different names according to the definition of relational schema, the functional dependency in between them is not treated as reflexive and should be in $\mathcal{F}$ as given in the example below. It should be noted that this kind of functional dependencies have a special importance since they are used to perfom equijoins on keys while inference, which will be discussed below. It is important to point out that the functional dependencies which were lost while building $\mathcal{R}$, are not in the scope of this paper; together with the key like behavior for a non-key attribute seen statistically through data for a schema.

\vspace{1.6mm}

For example,

\vspace{1.6mm}

$\mathcal{R} = \{USERS, LOGS\}$ as
\begin{multline*}
\begin{aligned}
\texttt{USERS = \{\underline{\emph{id\_users}}, }&\texttt{name\_users,} \\ &\texttt{surname\_users\}}
\end{aligned}
\end{multline*}
\begin{multline*}
\begin{aligned}
\texttt{LOGS = \{\underline{\emph{userid\_logs}}, }&\texttt{action\_logs,} \\ &\texttt{date\_logs\}}  
\end{aligned}
\end{multline*}

\texttt{userid\_logs} is a foreign key referencing \texttt{USERS(id\_users)}.

\vspace{3mm}

$\mathcal{F} = \left\lbrace 
\begin{aligned}
&\boldsymbol{(userid\_logs \rightarrow id\_users)}, \\
&\boldsymbol{(id\_users \rightarrow userid\_logs)}, \\ 
&(id\_users \rightarrow name\_users),\\
&(id\_users \rightarrow surname\_users), \\
&(userid\_logs \rightarrow action\_logs), \\
&(userid\_logs \rightarrow date\_logs) 
\end{aligned} \right\rbrace$

\vspace{3mm}

The functional dependencies given in bold expresses the dependency between original key and a foreign key. 

Using the join operation, the new relations can be obtained from the relations of logical schema. Similarly, the properties of functional dependencies can be used to generate new functional dependencies from the existing ones. Below, we define two closures for these two.

\begin{defn}[\textbf{$\boldsymbol{\mathcal{F}^+}$: Closure of $\boldsymbol{\mathcal{F}}$}]
Closure set of given $\mathcal{F}$ that can be obtained by using the properties of functional dependencies \cite{r23}, excluding reflexive and partial functional dependencies.
\end{defn}

\begin{defn}[\textbf{$\boldsymbol{\mathcal{R}^+}$: Closure of $\boldsymbol{\mathcal{R}}$}]
\textit{Closure of $\mathcal{R}$} is defined as $\mathcal{R}^+$, composing all probable relation schemas obtained by performing any database query over $\mathcal{R}$ for which join operations used in the query should only be \textit{equijoins on foreign keys} (named \textbf{meaningful join} thereafter) in order not to produce spurious tuples. 
\end{defn}

For example, by using the sample decomposed relations given in introduction section:

\vspace{1.6mm}

\texttt{STUDENT$_1$ = (name, surname, address, }

\hspace{2.4cm} \texttt{age, gender)}

\vspace{1.6mm}

\texttt{STUDENT$_2$ = (\underline{\emph{email}}, name, surname, }

\hspace{2.4cm} \texttt{address, age)}

\noindent
the query which produces spurious tuples can be given as:

\vspace{1.6mm}

\hspace{1cm} \texttt{SELECT *}

\hspace{1cm} \texttt{FROM STUDENT$_1$ s1, STUDENT$_2$ s2}

\hspace{1cm} \texttt{WHERE s1.name = s2.name}

\vspace{1.6mm}

\noindent
The join operation is an equijoin but not on keys, therefore spurious tuples are generated by associating different students having the same \texttt{name}.

Let $\boldsymbol{\mathcal{U_R}}$  be union of all attributes existing in all relations of $\mathcal{R}$ (i.e., $\mathcal{U_R}=\bigcup_{\mathcal{R}_i\in \mathcal{R}}\Set{x| x\in\mathcal{R}_i}$). Rather than defining which relations can be constructed from $\mathcal{R}$ as subsets of $\mathcal{U_R}$, we can specify the set of attributes that cannot be obtained together in $\mathcal{R}^+$ as follows:

\vspace{1.6mm}

\textbf{Property of $\boldsymbol{\mathcal{R}^+}$}:
An attribute set $\mathcal{A}$ cannot be subset of any derived relation schema in $\mathcal{R}^+$, if and only if, $\mathcal{A}$ cannot be a subset of existing relation schema or cannot be functionally dependent to any set of attributes in $\mathcal{U_R}$ so it becomes impossible to relate with equijoins on foreign keys by the definition of $\mathcal{R}^+$. Note that, according to Definition 2 if there is a foreign key relationship, then it is represented as functional dependency as $( ( A_i \rightarrow A_j ) \in\mathcal{F} )$. Since all meaningful joins can only be executed using this kind of functional dependencies, the following logical formula expresses that a set of attributes $\mathcal{A}$ cannot be a subset of any relation in $\mathcal{R}^+$ if and only if, there is no functional dependency relationship to $\mathcal{A}$ in $\mathcal{F}^+$ and there is no $\mathcal{R}$ containing $\mathcal{A}$.

\begin{equation}
\begin{split}
\forall\mathcal{R}_k\in\mathcal{R}^+, \forall \mathcal{A}\subseteq \mathcal{U_R}[ \mathcal{A} \not\subseteq\mathcal{R}_k \Leftrightarrow \forall \mathcal{R}_j\in \mathcal{R} ( \mathcal{A}\not\subseteq \mathcal{R}_j )  \\ \wedge  \forall\mathcal{A}_i\subseteq \mathcal{U_R} ( ( \mathcal{A}_i \rightarrow \mathcal{A} ) \notin\mathcal{F}^+ ) ] 
\end{split}
\end{equation}

As it can be seen from above definitions, there is a strong condition which says that in order not to be able to obtain a subset of attributes from a logical schema it should not be possible to perform a meaningful join. In order to perform meaningful join; foreign keys are used, and they correspond to functional dependencies. Therefore, we need the following definitions to represent these relationships.

\begin{defn}[\textbf{Set of Identifier Sets}]
The \textit{set of identifier set} of an attribute $\mathcal{\alpha}$  for a given $\mathcal{F}$, as $i^\mathcal{F}_\alpha$, is defined as follows:
\begin{equation}
i^\mathcal{F}_\alpha=\Set{x| x\not\supseteq \{\alpha\} \wedge (x \rightarrow \alpha)\in\mathcal{F}^+}
\end{equation}
\end{defn}

Each element of $i^\mathcal{F}_\alpha$ is also called as \textit{identifier set} of attibute $\alpha$.

\vspace{1.6mm}

For example,

\vspace{1.6mm}

\hspace{0.25cm} $\alpha = name$
 
\vspace{1mm}

\hspace{0.25cm} $\mathcal{F}^+ = \left\lbrace 
\begin{aligned}
&id\rightarrow name, \\ 
&id\rightarrow surname, \\
&id\rightarrow age, \\
&id\rightarrow email, \\
&email\rightarrow name, \\ 
&email \rightarrow surname \\
\end{aligned} \right\rbrace$ 

\vspace{1mm}

\hspace{0.25cm} $i^\mathcal{F}_\alpha = \{ \{id\}, \{email\}\}$

\vspace{1.6mm}

The definition could be simply extended for an attribute set $\mathcal{A}$ as follows:
\begin{equation}
\mathcal{I}^\mathcal{F}_\mathcal{A}=\Set{x| x\not\supseteq \mathcal{A} \wedge (x \rightarrow \mathcal{A})\in\mathcal{F}^+}
\end{equation}

These two definitions can be related as for an attribute set $\mathcal{A}$, as $\mathcal{I}^\mathcal{F}_\mathcal{A}$ contains the shared elements in $i^\mathcal{F}_\alpha$ for all $\alpha$ which are attributes in $\mathcal{A}$.
\begin{equation}
\mathcal{I}^\mathcal{F}_\mathcal{A}=\Set{x| \forall\alpha\in \mathcal{A}(x\in i^\mathcal{F}_\alpha)}
\end{equation}

\textbf{Identifiable Property}:
Each attribute of an \textit{identifiable set} (i.e, ($\mathcal{I}^\mathcal{F}_\mathcal{A}\neq\emptyset$)), should be in the same relational schema with at least one of its identifier set. In other words, for a $\mathcal{L=(R,F)}$;
\begin{equation}
\begin{split}
\forall\mathcal{A}\subseteq\mathcal{U_R}(&\mathcal{I}^\mathcal{F}_\mathcal{A} \neq\emptyset \Leftrightarrow \\ &\forall \alpha \in \mathcal{A}, \exists \mathcal{R}_i\in\mathcal{R}, \exists\mathcal{D}\in i^\mathcal{F}_\alpha ((\alpha \cup \mathcal{D}) \subseteq \mathcal{R}_i))
\end{split}
\end{equation}

The same property can be thought as it is impossible to identify an attribute in a logical schema if the attribute does not have any identifier set in its relation schema, which makes it impossible to discover other identifier sets by using meaningful joins. This issue is also a matter of database normalization however in this paper no assumption about the normal form of database has been done.

\begin{defn}[\textbf{Inferability}]
A set of attributes $\mathcal{A}_1\subseteq\mathcal{U_R}$ can be \textit{inferable} from a set of attributes $\mathcal{A}_2\subseteq\mathcal{U_R}$ for a given $\mathcal{L=(R,F)}$, shown as $\mathcal{A}_1\overset{\mathcal{F}}{\rightrightarrows}\mathcal{A}_2$, as defined below.
\begin{equation}
\forall\mathcal{A}_1\forall\mathcal{A}_2((\mathcal{A}_1\overset{\mathcal{F}}{\rightrightarrows}\mathcal{A}_2)\Leftrightarrow((\mathcal{A}_1\rightarrow\mathcal{A}_2)\in\mathcal{F}^+))
\end{equation}
\end{defn}

The definition of inferability is given by using closure set of functional dependencies since relation based key constraints may not be adequate as there may be functional dependencies among non-prime attributes in a schema which is not normalized.

\begin{defn}[\textbf{Inference of Association among a Set of Attributes}]
For a given $\mathcal{L=(R,F)}$, the inference of association among a set of attributes $\mathcal{A}\subseteq\mathcal{U_R}$ ,shown as $\mathcal{X}_\mathcal{L}(A)$, means that either $\mathcal{A}$ should be inferable from a subset of $\mathcal{U_R}$ or be subset of any existing relation schema in order to be associated. More formally:
\begin{equation}
\mathcal{X}_\mathcal{L}(A) \Leftrightarrow \exists\mathcal{A}_i\subseteq\mathcal{U_R}(\mathcal{A}_i\overset{\mathcal{F}}{\rightrightarrows}\mathcal{A})\vee\exists\mathcal{R}_i\in\mathcal{R}(\mathcal{A}\subseteq\mathcal{R}_i)
\end{equation}
\end{defn}

In this paper, the set of attributes are defined as to have \textbf{security dependency} among them if the inference of association among them should be prevented. In addition to that, the purpose of this paper is to inhibit the inference of association among a given subset of $\mathcal{U_R}$ with at least two attributes (each named a \textbf{security dependent set} thereafter) for a logical schema $\mathcal{L=(R,F)}$ by building a \textbf{secure logical schema}.

\begin{defn}[\textbf{Secure Logical Schema}]
A secure logical schema is a logical schema $\mathcal{L}_\mathcal{S}^{sec} = (\mathcal{R}, \mathcal{F})$ such that for a given set of security dependent sets $\mathcal{S}$, there should not be any relation in $\mathcal{R}^+$, containing the attributes of any set in $\mathcal{S}$. Formally:
\begin{equation}
\mathcal{L}_\mathcal{S}^{sec} \Leftrightarrow \forall \mathcal{S}_i \in \mathcal{S}, \nexists\mathcal{R}_i \in \mathcal{R}^+(\mathcal{S}_i \subseteq \mathcal{R}_i)
\end{equation}
\end{defn}

It should be emphasized that by the definition of $\mathcal{R}^+$, only meaningful joins are taken into consideration as queries should only have equijoins on foreign keys and by this way spurious tuples cannot be generated. 

By the definition of secure logical schema, it can be stated that the inference of association among attributes of each security dependent set is impossible with a secure logical schema since the attributes of any security dependent sets cannot be functionally dependent to any subset of attributes in logical schema (excluding reflexive and partial dependencies as given in Definition 2) or in the same relation as given as a theorem below.

\begin{thr}
The inference of association among attributes of each security dependent set cannot be performed in secure logical schemas; that is,
\begin{equation}
\forall \mathcal{S}_i \in \mathcal{S}(\neg\mathcal{X}_{\mathcal{L}_\mathcal{S}^{sec}}(\mathcal{S}_i))
\end{equation}
\end{thr}

\begin{proof}
The formal proof given in Appendix briefly states that in order to perform the disallowed inference, either the attributes should be in the same relation or a meaningful join should be done for an inter relation inference as both cases are impossible because of secure logical schema definition.
\end{proof}

\vspace{1.6mm}

The next step is to define transformation of a logical schema to a secure logical schema for given security dependency sets.

\begin{defn}[\textbf{Secure Decomposition}]
A \textit{secure decomposition} is decomposition of $\mathcal{L=(R,F)}$ according to the set of security dependency sets $\mathcal{S}$ to a new logical schema $\mathcal{L}_\mathcal{S}^{'} = (\mathcal{R'}, \mathcal{F'})$ , having the following features:

\begin{enumerate}

\item
Any attribute should not be lost after decomposition. In other words:
\begin{equation}
\mathcal{U_R}=\mathcal{U_{R'}}
\end{equation}

\item
The new set of functional dependencies should be subset of existing set of functional dependencies as there can't be any new functional dependency moreover a loss in existing functional dependencies is expected to inhibit the inference of associations among the elements of security dependent sets.
\begin{equation}
(\mathcal{F'}\subseteq\mathcal{F}) \wedge (\mathcal{F'}^+\subseteq\mathcal{F}^+)
\end{equation}

\item
Any of the decomposed relations should not be a superset of any security dependent set.
\begin{equation}
\forall \mathcal{S}_i \in \mathcal{S}, \nexists\mathcal{R}_i\in\mathcal{R'}(\mathcal{S}_i\subseteq\mathcal{R}_i)
\end{equation}

\item
Any of the attributes in a security dependent set should not coexist in the same decomposed relation with any of its identifier set.
\begin{equation}
\forall \mathcal{S}_i \in \mathcal{S}, \forall \mathcal{R}_i\in\mathcal{R'}, \forall \sigma \in \mathcal{S}_i, \nexists \tau \in i^\mathcal{F}_\sigma ((\{\sigma\}\cup\tau)\subseteq\mathcal{R}_i)
\end{equation}

\end{enumerate}
\end{defn}

The fourth property of secure decomposition is a strong requirement since it makes all the attributes in any security dependent set, uninferrable after the decomposition. It should be noted that this property is a requirement for a totally proactive solution. If any mechanism intends to have proactive and run time components together, then this requirement can be relaxed.

The aim of the secure decomposition is to transform a logical schema to a secure logical schema by the help of security dependency sets, which is given as a theorem below. 

\begin{thr}
If $\mathcal{L}_\mathcal{S}^{'} = (\mathcal{R'}, \mathcal{F'})$ is the logical schema obtained after performing secure decomposition to $\mathcal{L} = (\mathcal{R}, \mathcal{F})$ with the set of security dependency sets $\mathcal{S}$, then $\mathcal{L}_\mathcal{S}^{'}$ is a $\mathcal{L}_\mathcal{S}^{sec}$.
\end{thr}

\begin{IEEEproof}
The formal proof given in Appendix briefly states that in order not to be a secure logical schema, a security dependent set should be inferable or should be a part of an original relation. The former is impossible as the attributes of security dependent sets cannot be in the same relation with any of their identifier sets by the fourth property of the definition of secure decomposition. It is also impossible for the latter due to the third property of secure decomposition.

\end{IEEEproof}

\section{Decomposition Algorithm}

The main purpose of the decomposition algorithm is to achieve the secure decomposition (Definition 9) which is defined as resulting in a secure logical schema. In order to satisfy the goal, it is clear that the elements of each security dependent set should not be in the same sub-relation obtained after the decomposition of original relations. Furthermore, it should not be possible to meaningfully join two sub-relations containing securely dependent attributes separately. Below we define an algorithm which exhaustively generates all the subsets of the attributes of all relations and eliminates the ones that do not satisfy the conditions mentioned above. After that, it also eliminates redundant sub-relations.

\vspace{1.6mm}

Secure decomposition algorithm for the $\mathcal{L=(R,F)}$ with the given security dependencies set $\mathcal{S}$ is given in Algorithm 1.

\begin{algorithm}
\caption{Decomposition Algorithm}
\label{array-sum}

\begin{algorithmic}[1]
\Require
 \Statex $\mathcal{L}$: logical schema as $\mathcal{(R, F)}$,
 \Statex $\mathcal{S}$: set of security dependent sets for $\mathcal{L}$
\Ensure
 \Statex $\mathcal{P_R}$: set of maximal subsets of $\mathcal{R}$ according to $\mathcal{S}$
    \State \textbf{begin}
    \State $\mathcal{P_R} = \emptyset$
    \For {each $\mathcal{R}_x$ in $\mathcal{R}$}
        \State $\mathcal{P_R}_x$ = Power Set of $\mathcal{R}_x$
        \For {each $\mathcal{S}_i$ in $\mathcal{S}$}
			\For {each $\mathcal{Z}_i$ in $\mathcal{P_R}_x$}
				\If {$\mathcal{S}_i\subseteq\mathcal{Z}_i$}
        				\State remove $\mathcal{Z}_i$ from $\mathcal{P_R}_x$
    				\EndIf
				\For {each $\alpha$ in $\mathcal{S}_i$}
					\For {each $\lambda$ in $i^\mathcal{F}_\alpha$}
						\If {$(\{\alpha\}\cup\lambda)\subseteq\mathcal{Z}_i$}
							\State remove $\mathcal{Z}_i$ from $\mathcal{P_R}_x$
						\EndIf
					\EndFor
				\EndFor
			\EndFor
	   \EndFor
	   \For {each $\mathcal{V}_i$ in $\mathcal{P_R}_x$}
			\For {each $\mathcal{W}_i$ in $\mathcal{P_R}_x$}
				\If {$\mathcal{V}_i\subseteq\mathcal{W}_i$}
							\State remove $\mathcal{V}_i$ from $\mathcal{P_R}_x$
				\EndIf
			\EndFor
	   \EndFor
	   \State $\mathcal{P_R} = \mathcal{P_R} \cup \mathcal{P_R}_x$
    \EndFor
	\State \textbf{return} $\mathcal{P_R}$
	\State\textbf{end}
\end{algorithmic}
\end{algorithm}

The secure decomposing algorithm works as follows:

\vspace{1.6mm}

For all relational schemas in $\mathcal{R}$,

\vspace{1.6mm}
\begin{enumerate}
\item
Firstly, powerset of the a relational schema is generated, which is called as $\mathcal{P_R}_x$ in the algorithm (line (4)).
\item
Then, for each security dependency set in $\mathcal{S}$ (line (5)) each element of $\mathcal{P_R}_x$ (line (6)) is processed. The set is eliminated if: 
\begin{itemize}
\item
it contains all attributes of that security dependent set together (lines (7-9)), or,
\item
it contains one of the attributes of the security dependent set with the attribute's any identifier set together (lines (10-16))
\end{itemize}
\item
After that, among the remaining subsets; redundant ones (used for unnecessary sub-relations composed by other sub-relations) are also eliminated (lines (19-25)).
\end{enumerate}

The elimination strategy is aimed to create a secure logical schema. It is important to note that all of the work in this paper is concentrated on security dependent sets. Actually there may be some basic policy rules as a single attribute should not be accessed in any context and these basic cases can be easily handled with simple extensions to the algorithm. However it is left as a future work to define a complete mechanism.

\begin{thr}
Decomposition algorithm performs secure decomposition on given $\mathcal{L=(R,F)}$ for a given $\mathcal{S}$. 
\end{thr}

\begin{IEEEproof}
For the proof, we should revisit the properties of secure decomposition given in Definition 9. Each explanation is used in the same item number with the definition.
\begin{enumerate}
\item
Any attribute cannot be lost after decomposition algorithm as the algorithm cannot remove one element subsets of each relation since;
\begin{itemize}
\item
Security dependency sets should have at least two elements by its definition and cannot be included by a one element subset.
\item
An attribute in a security dependent set cannot be its identifier as reflexive functional dependencies are excluded in the definition of $\mathcal{F}^+$ and the set of identifier set thereby. So again, at least two element set (attribute and its identifier) cannot be subset of one element subset.
\end{itemize}
\item
Any new functional dependencies cannot be presented besides some existing ones may be lost because some subsets of each relation schema are eliminated.
\item
Subsets containing security dependent sets are eliminated (lines (7-9)).
\item
All subsets containing an attribute from a security dependent set and it's any identifier set are eliminated (lines (10-16)).
\end{enumerate}
\end{IEEEproof}

The following parameters of $\mathcal{L=(R,F)}$ and $\mathcal{S}$ affect the performance of decomposition algorithm.

\begin{itemize}
\item
$\pi$ : $\#relations \in \mathcal{R}$

\vspace{1.6mm}

\item
$\epsilon$ : $max_{\mathcal{R}_i \in \mathcal{R}}\{\left\vert\mathcal{R}_i\right\vert\}$

\vspace{1.6mm}

\item
$\eta$ : $max_{\mathcal{S}_i \in \mathcal{S}}\{\left\vert\mathcal{S}_i\right\vert\}$

\vspace{1.6mm}

\item
$\mu$ : $max_{\mathcal{\alpha} \in \mathcal{S}_i(\mathcal{S}_i \in \mathcal{S})}\{\left\vert i^\mathcal{F}_{\alpha} \right\vert\}$
\end{itemize}

The algorithm works at a cost of $\mathcal{O}(\pi \cdot 2^\epsilon \cdot \eta \cdot \mu)$. The problem is relevant to generating maximal independent sets problem \cite{r12} in an undirected graph, in which attributes can be thought as vertices and dependencies as edges. In \cite{r12} it is shown that generating maximal independent sets is NP-Hard. 

It is important to note that being a proactive solution, the exponential complexity of decomposition algorithm is not a critical problem since it is only executed once as preprocessing phase. 

Another point for the decomposition is that, this may lead to relations with no key and because of that reason, duplicate rows may occur in views. Handling mechanism for duplicate rows can change up to the implementation strategy, and keyless relations are not a problem since they are a fact of anonymity.

\section{Real Life Example}

Consider a retail store database with a logical schema $\mathcal{L=(R,F)}$ as the following three relations in $\mathcal{R}$:

\vspace{1.6mm}

$\mathcal{R} = \{CUSTOMER, PRODUCT, BUY\}$ 

\vspace{1.6mm}

$CUSTOMER$ table is used for storing customer details, $PRODUCT$ table is for product information and $BUY$ relation stores the purchase transactions of customers. It should be noted that many tables and attributes that can be useful for a retail store has been omitted not to overcomplicate the example.

\vspace{1.6mm}

The relation schemas are given as:

\vspace{3mm}

\hspace{0.25cm} $CUSTOMER = \left\lbrace 
\begin{aligned}
&\underline{customerId(cid)}, \\
&name, \\ 
&surname, \\
&phoneNumber(pNo), \\
&address,  \\
&age, \\
&gender 
\end{aligned} \right\rbrace$

\vspace{3mm}

\hspace{0.6cm} $PRODUCT = \left\lbrace 
\begin{aligned}
&\underline{productId(pid)}, \\
&name, \\ 
&model, \\
&year, \\
&price
\end{aligned} \right\rbrace$

\vspace{3mm}

\hspace{1.7cm} $BUY = \left\lbrace 
\begin{aligned}
&\underline{customerId}, \\
&\underline{productId}, \\
&\underline{date}, \\ 
&quantity
\end{aligned} \right\rbrace$

\vspace{3mm}

$\mathcal{F} = \left\lbrace 
\begin{aligned}
&cid\_customer \rightarrow name\_customer, \\ 
&cid\_customer \rightarrow surname\_customer, \\
&cid\_customer \rightarrow pNo\_customer, \\
&cid\_customer \rightarrow address\_customer, \\
&cid\_customer \rightarrow age\_customer, \\ 
&cid\_customer \rightarrow gender\_customer, \\
&pid\_product \rightarrow name\_product, \\
&pid\_product \rightarrow model\_product, \\
&pid\_product \rightarrow year\_product, \\
&pid\_product \rightarrow price\_product, \\
&\left( 
\begin{aligned}
&cid\_buy \\
&pid\_buy \\
&date\_buy 
\end{aligned}
\right)
\rightarrow quantity\_buy, \\
& \boldsymbol{cid\_customer \rightarrow cid\_buy}, \\
& \boldsymbol{cid\_buy \rightarrow cid\_customer}, \\
& \boldsymbol{pid\_product \rightarrow pid\_buy}, \\
& \boldsymbol{pid\_buy \rightarrow pid\_product}
\end{aligned} \right\rbrace$ 

\vspace{3mm}

In addition to these, $\mathcal{F}$ is defined as a set of functional dependencies. As given in Definition 2, $\mathcal{F}$ includes dependencies for the foreign keys which are given in bold. It is important that each foreign key based functional dependency always exists with its symmetric pair, since the foreign and its relevant key are the same attribute.

Sample tuples of these three relations are illustrated in Tables I, II, and III respectively. The fields $customerId$ and $productId$ are the keys of the $CUSTOMER$ and $PRODUCT$ relations respectively, and they form a composite key in the $BUY$ relation together with the $date$ attribute. 
 
{\renewcommand{\arraystretch}{2}%
\begin{table}[!h]
\centering
    \caption{$CUSTOMER$ Relation Sample Data} \begin{small}
    \begin{tabular}{|c|c|c|c|c|c|c|}
    \hline
    {\bfseries cid} & {\bfseries  name} & {\bfseries surname} & {\bfseries pNo}  & {\bfseries address}  & {\bfseries age}  & {\bfseries gender}          \\
    \hline
    1         	& John		& Doe		& 5555555		& NYC	& 21	& M \\
        
    \hline
    2         	& Mary		& Doe 		& 6666666		& NYC	& 28	& F \\
        
    \hline
    3        	& Mary		& White	& 7777777		& York 	& 28	& F \\

    \hline
    \end{tabular}
    \end{small} 
\end{table}}

{\renewcommand{\arraystretch}{2}%
\begin{table}[!h]
\centering
    \caption{$PRODUCT$ Relation Sample Data} \begin{small}
    \begin{tabular}{|c|c|c|c|c|}
    \hline
    {\bfseries pid} & {\bfseries  name} & {\bfseries model} & {\bfseries year}  & {\bfseries price}  \\
    \hline
    1         	& PS		& 3		& 2012		& 599,00 	\\
        
    \hline
    2         	& XBOX	& 360 	& 2013		& 799,00	\\
        
    \hline
    3        	& PS		& 4		& 2014		& 899,00 	\\

    \hline
    \end{tabular}
    \end{small} 
\end{table}}

{\renewcommand{\arraystretch}{2}%
\begin{table}[!h]
\centering
    \caption{$BUY$ Relation Sample Data} \begin{small}
    \begin{tabular}{|c|c|c|c|}
    \hline
    {\bfseries cid} & {\bfseries  pid} & {\bfseries date} & {\bfseries quantity}  \\
    \hline
    1         	& 2		& 20140701-16:28:47	& 1		\\
        
    \hline
    1         	& 3		& 20140702-19:07:11	& 2		\\
        
    \hline
    3        	& 3		& 20140703-12:30:05	& 2		\\

    \hline
    \end{tabular}
    \end{small} 
\end{table}}

It should be mentioned that some non-key attributes may be used as a pseudo-key to recover the original relation after decomposition. However, that is due to the distribution of data values and it is not in the scope of our work.

As an application, consider the business development department of the retail store which checks the $BUY$ relation and then investigates the correlation among purchases and customer characteristics as $gender$ and $age$. However, there can be a malicious worker in the department who can share customer details with other stores as customer access information. The attributes $address$ and $phoneNumber$ may be used to present promotions to customer by a different store. In addition to that, $customerId$ need not to be related with $address$ and $phoneNumber$ since the department's major objective is to use $age$ and $gender$ attributes. To prevent this situation, following security dependency sets should be defined.

\vspace{2mm}

\hspace{0.25cm} $\mathcal{S} = \left\lbrace 
\begin{aligned}
&\left\lbrace 
\begin{aligned}
&cid\_buy, \\
&address\_customer
\end{aligned}
\right\rbrace  \\
&\left\lbrace 
\begin{aligned}
&cid\_buy, \\
&pNo\_customer
\end{aligned}
\right\rbrace \\
&\left\lbrace 
\begin{aligned}
&address\_customer, \\
&age\_customer
\end{aligned}
\right\rbrace  \\
&\left\lbrace 
\begin{aligned}
&address\_customer, \\
&gender\_customer
\end{aligned}
\right\rbrace \\
&\left\lbrace 
\begin{aligned}
&pNo\_customer, \\
&age\_customer
\end{aligned}
\right\rbrace \\
&\left\lbrace 
\begin{aligned}
&pNo\_customer, \\
&gender\_customer
\end{aligned}
\right\rbrace
\end{aligned} \right\rbrace$ 

\vspace{2mm}

It is important to note that to give a single security dependency set as 

\vspace{2mm}

\hspace{0.25cm} $\mathcal{S}_{faulty} = \left\lbrace 
\begin{aligned}
&\left\lbrace 
\begin{aligned}
&cid\_buy, \\
&age\_customer, \\
&gender\_customer, \\
&address\_customer, \\
&pNo\_customer
\end{aligned}
\right\rbrace 
\end{aligned} \right\rbrace$ 

\vspace{2mm}

\noindent
is an example of faulty definition since decomposition will not inhibit the inference of association among any two of the attributes.

\vspace{2mm}

To figure out the identifiers, dependencies in $\mathcal{F}^+$ should be reproduced.

\vspace{1.6mm}

\begin{enumerate}
\item
$cid\_buy \rightarrow cid\_customer \in \mathcal{F}$

\vspace{2mm}

\item
$cid\_customer \rightarrow  \left( \begin{aligned} &age\_customer, \\ &gender\_customer, \\  &address\_customer, \\ &pNo\_customer \end{aligned} \right) \in \mathcal{F}$

\vspace{2mm}

\item
$cid\_buy \rightarrow  \left( \begin{aligned} &age\_customer, \\ &gender\_customer, \\  &address\_customer, \\ &pNo\_customer  \end{aligned} \right) \in \mathcal{F}^+$ (transitive property (1, 2))
\end{enumerate}

\vspace{2mm}

As a result of the dependencies in $\mathcal{F}^+$, identifier sets are constructed as follows:

\vspace{2mm}

\begin{itemize}
\item
$i^\mathcal{F}_{age\_customer} = \{\{cid\_buy\}, \{cid\_customer\}\}$

\vspace{2mm}

\item
$i^\mathcal{F}_{gender\_customer} = \{\{cid\_buy\}, \{cid\_customer\}\}$

\vspace{2mm}

\item
$i^\mathcal{F}_{address\_customer} = \{\{cid\_buy\}, \{cid\_customer\}\}$

\vspace{2mm}

\item
$i^\mathcal{F}_{pNo\_customer} = \{\{cid\_buy\}, \{cid\_customer\}\}$

\vspace{2mm}

\item
$i^\mathcal{F}_{cid\_buy} = \{\{cid\_customer\}\}$
\end{itemize}

\vspace{2mm}

Decomposition Algorithm will work as follows for $CUSTOMER$ relation:

\vspace{2mm}

\begin{enumerate}
\item
All subsets of attributes in $CUSTOMER$ relation will be generated ($2^7 = 128$ subsets).
\item
The subsets which will be removed firstly, are the ones containing securely dependent attributes such as:

\vspace{0.6mm}

\begin{itemize}
\item
$\left\lbrace 
cid\_buy, address\_customer
\right\rbrace$

\vspace{0.8mm}

\item
$\left\lbrace 
cid\_buy, pNo\_customer
\right\rbrace$

\vspace{0.8mm}

\item
$\left\lbrace 
address\_customer, age\_customer
\right\rbrace$

\vspace{0.8mm}

\item
$\left\lbrace 
address\_customer, gender\_customer
\right\rbrace$

\vspace{0.8mm}

\item
$\left\lbrace 
pNo\_customer, age\_customer
\right\rbrace$

\vspace{0.8mm}

\item
$\left\lbrace 
pNo\_customer, gender\_customer
\right\rbrace$

\end{itemize}

\vspace{0.6mm}

\item
The subsets which will be removed next, are the ones containing an attribute in a security dependent set with its identifier such as:

\vspace{0.6mm}

\begin{itemize}
\item
$\left\lbrace 
cid\_buy, cid\_customer
\right\rbrace$

\vspace{0.8mm}

\item
$\left\lbrace 
cid\_buy, address\_customer
\right\rbrace$

\vspace{0.8mm}

\item
$\left\lbrace 
cid\_buy, pNo\_customer
\right\rbrace$

\vspace{0.8mm}

\item
$\left\lbrace 
cid\_buy, age\_customer
\right\rbrace$

\vspace{0.8mm}

\item
$\left\lbrace 
cid\_buy, gender\_customer
\right\rbrace$

\vspace{0.8mm}

\item
$\left\lbrace 
cid\_customer, address\_customer
\right\rbrace$

\vspace{0.8mm}

\item
$\left\lbrace 
cid\_customer, pNo\_customer
\right\rbrace$

\vspace{0.8mm}

\item
$\left\lbrace 
cid\_customer, age\_customer
\right\rbrace$

\vspace{0.8mm}

\item
$\left\lbrace 
cid\_customer, gender\_customer
\right\rbrace$

\end{itemize}

\vspace{0.6mm}

\item
Unnecessary subsets which are composed by other subsets, are also removed.
\end{enumerate}

\vspace{2mm}

After decomposition algorithm is applied to the $CUSTOMER$ relation, the following sub-relations are going to be obtained.

\vspace{2mm}

\hspace{0.25cm} $CUSTOMER_1 = \left\lbrace 
\begin{aligned}
&\underline{cid}, \\
&name, \\ 
&surname
\end{aligned} \right\rbrace$

\vspace{2mm}

\hspace{0.25cm} $CUSTOMER_2 = \left\lbrace 
\begin{aligned}
&name, \\ 
&surname, \\
&pNo, \\
&address
\end{aligned} \right\rbrace$

\vspace{2mm}

\hspace{0.25cm} $CUSTOMER_3 = \left\lbrace 
\begin{aligned}
&name, \\ 
&surname, \\
&gender, \\
&age
\end{aligned} \right\rbrace$

\vspace{2mm}

It should be noted that $name$ and $surname$ fields are not keys, if they may perform key-like behavior, then $\mathcal{S}$ should be arranged to compose new security dependent sets considering $name$ and $surname$. However it is not in the scope of this paper to identify key-like behaviour.

{\renewcommand{\arraystretch}{2}%
\begin{table}[!h]
\centering
    \caption{$CUSTOMER_1$ Relation Sample Data} \begin{small}
    \begin{tabular}{|c|c|c|}
    \hline
    {\bfseries cid} & {\bfseries  name} & {\bfseries surname}  \\
    \hline
    1         	& John		& Doe		\\
        
    \hline
    2         	& Mary		& Doe 		 \\
        
    \hline
    3        	& Mary		& White	 \\

    \hline
    \end{tabular}
    \end{small} 
\end{table}}

{\renewcommand{\arraystretch}{2}%
\begin{table}[!h]
\centering
    \caption{$CUSTOMER_2$ Relation Sample Data} \begin{small}
    \begin{tabular}{|c|c|c|c|}
    \hline
    {\bfseries  name} & {\bfseries surname} & {\bfseries pNo}  & {\bfseries address} \\
    \hline
    John		& Doe		& 5555555		& NYC	 \\
        
    \hline
    Mary		& Doe 		& 6666666		& NYC	 \\
        
    \hline
    Mary		& White	& 7777777		& York 	 \\

    \hline
    \end{tabular}
    \end{small} 
\end{table}}

{\renewcommand{\arraystretch}{2}%
\begin{table}[!h]
\centering
    \caption{$CUSTOMER_3$ Relation Sample Data} \begin{small}
    \begin{tabular}{|c|c|c|c|}
    \hline
    {\bfseries  name} & {\bfseries surname} & {\bfseries age}  & {\bfseries gender}          \\
    \hline
    John		& Doe		& 21	& M \\
        
    \hline
    Mary		& Doe 		& 28	& F \\
        
    \hline
    Mary		& White	& 28	& F \\

    \hline
    \end{tabular}
    \end{small} 
\end{table}}

Three decomposed relations are given in Tables IV, V and VI respectively and no decomposition has been made to $BUY$ and $PRODUCT$ relations. In order to be able to find the number of customer at each age, $CUSTOMER_3$ relation can be used. The query is given below:

\vspace{2mm}

{\hspace{0.3cm} \texttt{QUERY-1:}

 \hspace{0.6cm} \texttt{SELECT age, COUNT(*) as 'count'}

\hspace{0.6cm} \texttt{FROM CUSTOMER$_3$}

\hspace{0.6cm} \texttt{GROUP BY age} 

\vspace{2mm}

Moreover, $name$ and $surname$ of the customers who has purchased \textit{"PS"} can be found as;

\vspace{2mm}

\hspace{0.3cm} \texttt{QUERY-2:}

\hspace{0.6cm} \texttt{SELECT c1.name, c1.surname}

\hspace{0.6cm} \texttt{FROM CUSTOMER$_1$ c1, BUY b, PRODUCT p}

\hspace{0.6cm} \texttt{WHERE c1.customerId = b.customerId} 

\hspace{1.4cm} \texttt{and c1.productId = p.productId}

\hspace{1.4cm} \texttt{and b.name = 'PS' }

\vspace{2mm}

The result of these two queries are given in Tables VII and VIII. It should be noted that, any attempt to infer the association among the attributes of each security dependent set cannot be done since the decomposed relations containing the securely dependent attributes cannot be joined on shared identifiers.

{\renewcommand{\arraystretch}{2}%
\begin{table}[!h]
\centering
    \caption{$QUERY_1$ Resulting Relation} \begin{small}
    \begin{tabular}{|c|c|}
    \hline
    {\bfseries  age} & {\bfseries count}         \\
    \hline
    21	& 1 \\
        
    \hline
    28	& 2 \\

    \hline
    \end{tabular}
    \end{small} 
\end{table}}

{\renewcommand{\arraystretch}{2}%
\begin{table}[!h]
\centering
    \caption{$QUERY_2$ Resulting Relation} \begin{small}
    \begin{tabular}{|c|c|}
    \hline
    {\bfseries  name} & {\bfseries surname}         \\
    \hline
    John	& Doe \\
        
    \hline
    Mary	& White \\

    \hline
    \end{tabular}
    \end{small} 
\end{table}}

\section{Future Work}

Being the first attempt in the literature to formalize proactive context dependent attribute based access control, the paper is touching on many different applicational and theoretical researches about the subject. Firstly, the way of usage of this work in the access control mechanisms of database management systems should be investigated. The proposed solution could work as a part of a trusted access control system of a database which has proactive and run-time components. Moreover, security dependent sets are not the only policy rule for database access control, it is also left as a future work for the model and algorithm to be expanded to satisfy all attribute and row based policy rules. By this generalization, performance and scalability should be investigated for new rule types (attribute or row based) about being "proactive", as some modules may be designed to execute in run time. Also it should be noted that, the last property of secure decomposition given in Definition 9 should be revisited if the mechanism employs some run-time components. The property may be relaxed with some run-time work, as all attributes existing in a security dependent set should not be made unidentifiable in advance, since this decision can be made depending on the query during run time. To sum up, the proposed secure decomposition solution in this paper performs in a proactive manner totally, and if future researches will present some run-time work into the mechanism, then this property of secure decomposition will be revisited for relaxation by also taking how performance affected into consideration.

Furthermore, access control strategy should be extended to data insertion or update policies according to the security dependencies. For this issue, implementation alternatives of the proposed solution should be defined, and data modification strategy with proactive and run-time components should be suggested accordingly. In addition to that, the changes done to the secure logical schema during run-time are another challenge for that paper which should be dealt in future. The modification may be an alternation of any relational schema, meanwhile, the functional dependencies may also be changed according to the relational schema. These changes will be expected to cause modifications to security dependent sets and it should be clearly examined how to handle these types of modifications for the proposed model.

Lastly, some practical applications of this work can be proposed in future. One of them is an applicational module, which inputs current external logical schema with security dependent sets of each user role and investigates whether the external schema has any security leak for any defined roles. This application can be useful in current database management systems for a verification of access control mechanism. Beside this, another application can be built to determine the allowed inferences of associations among attributes for any roles, since it is important to clarify what is allowed as what is inhibited in an access control strategy. This application can be used for a cross control for the design and requirements of an application.

\section{Conclusion}

The given theorem, algorithm and examples in this paper aims a construct proactive context dependent attribute based security mechanism schema for database users, using given security dependent sets. The main objective in this work is to prevent inference of association of the attributes in each security dependent set, and this is accomplished by performing a secure decomposition which transforms the relevant logical schema to a secure logical schema for which it is proven to be impossible to infer association among any security dependent set. Furthermore, an algorithm is proposed and proven to perform secure decomposition. It should be noted that all work in this paper are about building the external schema of the database according to the given logical schema (including relational schemas and functional dependencies) and security dependent sets, and it can be implemented independently from conceptual and physical model. As a result, different external schemas for all different roles of users in database has been achieved, and each role can access to the database through a different view from the point of security. By this work, granularity problem for access control methods for databases has been addressed and a formal context dependent proactive access control method has been proposed to be used in access control mechanisms of database management systems.

\bibliographystyle{IEEEtran}

\bibliography{IEEEabrv,IEEEexample}

\section*{Appendices}

The step by step proofs are given below with a brief description of each step.

\subsection*{Proof For Theorem-1}
\begin{proof}

\begin{enumerate}

\item
Assuming $\mathcal{L}$ as a secure logical schema, the formula given in (8)  should be satisfied.

\vspace{2mm}

$\forall \mathcal{S}_i ((\mathcal{S}_i \in \mathcal{S}) \Rightarrow \forall \mathcal{R}_j(\mathcal{R}_j \in \mathcal{R}^+ \Rightarrow \mathcal{S}_i \not\subseteq \mathcal{R}_j))$

\vspace{2.8mm}

\item
Let $\mathcal{S}_i$ be a security dependent set for $\mathcal{L}$. 

\vspace{2mm}

$\mathcal{S}_i \in \mathcal{S}$

\vspace{2.8mm}

\item
Line (1) can be instantiated by using $\mathcal{S}_i$.

\vspace{2mm}

$(\mathcal{S}_i \in \mathcal{S}) \Rightarrow \forall \mathcal{R}_j(\mathcal{R}_j \in \mathcal{R}^+ \Rightarrow \mathcal{S}_i \not\subseteq \mathcal{R}_j)$

\vspace{2.8mm}

\item
When modus ponens is applied using lines (2) and (3).

\vspace{2mm}

$\forall \mathcal{R}_j(\mathcal{R}_j \in \mathcal{R}^+ \Rightarrow \mathcal{S}_i \not\subseteq \mathcal{R}_j)$

\vspace{2.8mm}

\item
Bu using the formula (1) of the property of $\mathcal{R}^+$,  $\mathcal{S}_i$ should not be element of any existing relation $\mathcal{R}$ and there should not exist any attibute set to which $\mathcal{S}_i$ is functionally dependent since according to line (4), $\mathcal{S}_i$ is not a part of any relation in $\mathcal{R}^+$ so any new relation composing $\mathcal{S}_i$ should not be produced by meaningful joins. 

\vspace{2mm}

$\forall \mathcal{R}_k(\mathcal{R}_k \in \mathcal{R} \Rightarrow \mathcal{S}_i \not\subseteq \mathcal{R}_k) \wedge \forall \mathcal{A}_l ( \mathcal{A}_l  \subseteq \mathcal{U_R} \Rightarrow (\mathcal{A}_l  \rightarrow \mathcal{S}_i)\not\in\mathcal{F}^+)$

\vspace{2.8mm}

\item
Assume that the inference of association among the attributes in $\mathcal{S}_i$ can be done. This assumption is the negation of the theorem 1, so proof by contradiction starts here.

\vspace{2mm}

$\mathcal{X}_\mathcal{L}(\mathcal{S}_i)$

\vspace{2.8mm}

\item
According to line (6), the formula (7) states that $\mathcal{S}_i$ should be inferable or subset of any existing relation in $\mathcal{R}$.

\vspace{2mm}

$\exists \mathcal{A}_n (\mathcal{A}_n \overset{\mathcal{F}}{\rightrightarrows} \mathcal{S}_i \wedge \mathcal{A}_n \subseteq \mathcal{U_R}) \vee \exists \mathcal{R}_o ( \mathcal{R}_o \in \mathcal{R} \wedge \mathcal{S}_i \subseteq \mathcal{R}_o)$

\vspace{2.8mm}

\item
In order to contradict the $\vee$ expression in line (7), both sides of $\vee$ should be contradicted. Accordingly, the first assumption is given below as $\mathcal{S}_i$ should be inferable.

\vspace{2mm}

$\exists \mathcal{A}_n (\mathcal{A}_n \overset{\mathcal{F}}{\rightrightarrows} \mathcal{S}_i \wedge \mathcal{A}_n \subseteq \mathcal{U_R})$

\vspace{2.8mm}

\item
Let the expression in line (8) be instantiated using bound variable $\mathcal{A}_n$ denoting a attribute set which infers $\mathcal{S}_i$.

\vspace{2mm}

$\mathcal{A}_n \overset{\mathcal{F}}{\rightrightarrows} \mathcal{S}_i \wedge \mathcal{A}_n \subseteq \mathcal{U_R}$

\vspace{2.8mm}

\item
First $\wedge$ instantiation using line (9).

\vspace{2mm}

$\mathcal{A}_n \overset{\mathcal{F}}{\rightrightarrows} \mathcal{S}_i$

\vspace{2.8mm}

\item
Second $\wedge$ instantiation using line (5).

\vspace{2mm}

$\forall \mathcal{A}_l ( \mathcal{A}_l  \subseteq \mathcal{U_R} \Rightarrow (\mathcal{A}_l  \rightarrow \mathcal{S}_i)\not\in\mathcal{F}^+)$

\vspace{2.8mm}

\item
The universal quantifier in line (11) is instantiated using $\mathcal{A}_n$

\vspace{2mm}

$\mathcal{A}_n  \subseteq \mathcal{U_R} \Rightarrow (\mathcal{A}_n  \rightarrow \mathcal{S}_i)\not\in\mathcal{F}^+$

\vspace{2.8mm}

\item
Second $\wedge$ instantiation using line (9).

\vspace{2mm}

$\mathcal{A}_n  \subseteq \mathcal{U_R}$

\vspace{2.8mm}

\item
When modus ponens is applied using lines (12) and (13).

\vspace{2mm}

$(\mathcal{A}_n  \rightarrow \mathcal{S}_i)\not\in\mathcal{F}^+$

\vspace{2.8mm}

\item
Line (14) can be used to perform modus ponens to the contrapositive of formula (6).

\vspace{2mm}

$\neg(\mathcal{A}_n \overset{\mathcal{F}}{\rightrightarrows} \mathcal{S}_i)$, (14) using formula 6

\vspace{2.8mm}

\item
Lines (10) and (15) are leading to a contradiction.

\vspace{2mm}

$\perp$

\vspace{2.8mm}

\item
First assumption of $\vee$ expression in line (7) in line (8) has been contradicted. Next, the second assumption is given below as $\mathcal{S}_i$ should be a part of an existing relation.

\vspace{2mm}

$\exists \mathcal{R}_o ( \mathcal{R}_o \in \mathcal{R} \wedge \mathcal{S}_i \subseteq \mathcal{R}_o)$

\vspace{2.8mm}

\item
First $\wedge$ instantiation using line (5).

\vspace{2mm}

$\forall \mathcal{R}_k(\mathcal{R}_k \in \mathcal{R} \Rightarrow \mathcal{S}_i \not\subseteq \mathcal{R}_k)$

\vspace{2.8mm}

\item
The existential quantifier in line (17) is instantiated using bounded variable $\mathcal{R}_o$.

\vspace{2mm}

$\mathcal{R}_o \in \mathcal{R} \wedge \mathcal{S}_i \subseteq \mathcal{R}_o$, let $\mathcal{R}_o$ be a bound variable for $\mathcal{R}_o$

\vspace{2.8mm}

\item
Line (18) can be instatiated again by using $\mathcal{R}_o$.

\vspace{2mm}

$\mathcal{R}_o \in \mathcal{R} \Rightarrow \mathcal{S}_i \not\subseteq \mathcal{R}_o$

\vspace{2.8mm}

\item
First $\wedge$ instantiation using line (19).

\vspace{2mm}

$\mathcal{R}_o \in \mathcal{R}$

\vspace{2.8mm}

\item
Second $\wedge$ instantiation using line (19).

\vspace{2mm}

$\mathcal{S}_i \subseteq \mathcal{R}_o$

\vspace{2.8mm}

\item
When modus ponens is applied using lines (20) and (21).

\vspace{2mm}

$\mathcal{S}_i \not\subseteq \mathcal{R}_o$, (20, 21)

\vspace{2.8mm}

\item
Lines (22) and (23) are leading to a contradiction for line (17).

\vspace{2mm}

$\perp$

\vspace{2.8mm}

\item
Lines (16) and (24) are leading to a contradiction for line (7) which means that it is impossible to make an inference of association among the attributes in $\mathcal{S}_i$.

\vspace{2mm}

$\perp$

\vspace{2.8mm}

\item
End of proof by contradiction is reached, hence the theorem holds.

\vspace{2mm}

$\neg\mathcal{X}_\mathcal{L}(\mathcal{S}_i)$\end{enumerate}\end{proof}

\subsection*{Proof For Theorem-2}
\begin{IEEEproof}

\begin{enumerate}

\item
Assume that $\mathcal{L'}$ is a secure logical schema, then it should satisfy the following property given in formula (9).

\vspace{2mm}

$\forall \mathcal{S}_i (\mathcal{S}_i  \in \mathcal{S}) \Rightarrow \neg\mathcal{X}_\mathcal{L'}(\mathcal{S}_i)$

\vspace{3mm}

\item
Let $\mathcal{S}_i$ be a security dependent set for $\mathcal{L'}$. 

\vspace{2mm}

$\mathcal{S}_i \in \mathcal{S}$

\vspace{3mm}

\item
Line (1) can be instantiated by using $\mathcal{S}_i$.

\vspace{2mm}

$(\mathcal{S}_i  \in \mathcal{S}) \Rightarrow \neg\mathcal{X}_\mathcal{L'}(\mathcal{S}_i)$

\vspace{3mm}

\item
When modus ponens is applied using lines (2) and (3).

\vspace{2mm}

$\neg\mathcal{X}_\mathcal{L'}(\mathcal{S}_i)$

\vspace{3mm}

\item
Proof by contradiction begins by assuming negation of line (4) as if $\mathcal{L'}$ is not a secure logical schema, then inference of association among the attributes of a security dependent set  $(\mathcal{S}_i)$ should be possible according to theorem (1).

\vspace{2mm}

$\mathcal{X}_\mathcal{L'}(\mathcal{S}_i)$

\vspace{3mm}

\item
According to line (5), the formula (7) states that $\mathcal{S}_i$ should be inferable or subset of any existing relation in $\mathcal{R}$.

\vspace{2mm}

$\exists \mathcal{A}_n (\mathcal{A}_n \overset{\mathcal{F'}}{\rightrightarrows} \mathcal{S}_i \wedge \mathcal{A}_n \subseteq \mathcal{U_{R'}}) \vee \exists \mathcal{R}_o ( \mathcal{R}_o \in \mathcal{R'} \wedge \mathcal{S}_i \subseteq \mathcal{R}_o)$

\vspace{3mm}

\item
Let the expression in line (6) be instantiated using bound variable $\mathcal{A}_n$ denoting a attribute set which infers $\mathcal{S}_i$ and $\mathcal{R}_o$ denoting the relation which contains $\mathcal{S}_i$.

\vspace{2mm}

$(\mathcal{A}_n \overset{\mathcal{F'}}{\rightrightarrows} \mathcal{S}_i \wedge \mathcal{A}_n \subseteq \mathcal{U_{R'}}) \vee ( \mathcal{R}_o \in \mathcal{R'} \wedge \mathcal{S}_i \subseteq \mathcal{R}_o)$

\vspace{3mm}

\item
In order to contradict the $\vee$ expression in line (7), both sides of $\vee$ should be contradicted. Accordingly, the first assumption is given below.

\vspace{2mm}

$(\mathcal{A}_n \overset{\mathcal{F'}}{\rightrightarrows} \mathcal{S}_i \wedge \mathcal{A}_n \subseteq \mathcal{U_{R'}})$

\vspace{3mm}

\item
First $\wedge$ instantiation using line (8).

\vspace{2mm}

$\mathcal{A}_n \overset{\mathcal{F'}}{\rightrightarrows} \mathcal{S}_i$, (8)

\vspace{3mm}

\item
Formula (6) states that  $\mathcal{S}_i$ should be functionally dependent to an attribute set when the statement in line (9) exists.

\vspace{2mm}

$(\mathcal{A}_n \rightarrow \mathcal{S}_i) \in \mathcal{F'}^+$

\vspace{3mm}

\item
Using formula (11), the dependency in line (10) can be transformed as below.

\vspace{2mm}

$(\mathcal{A}_n \rightarrow \mathcal{S}_i) \in \mathcal{F}^+$

\vspace{3mm}

\item
$\mathcal{S}_i$ in line (11) cannot be contained by $\mathcal{A}_n$ as partial functional dependencies are excluded in definition of $\mathcal{F}^+$ in Definition (3).

\vspace{2mm}

$\mathcal{A}_n \not\supseteq \mathcal{S}_i$

\vspace{3mm}

\item
Using lines (11) and (12), it can be stated that $\mathcal{A}_n$ is an identifier set for $\mathcal{S}_i$ according to formula (3).

\vspace{2mm}

$\mathcal{A}_n \in \mathcal{I}^\mathcal{F}_ {\mathcal{S}_i}$

\vspace{3mm}

\item
If $\mathcal{L'}$ is a secure logical schema, then formula (13) should be satisfied.

\vspace{2mm}

$\forall \mathcal{S}_i \in \mathcal{S}, \forall \mathcal{R}_i\in\mathcal{R'}, \forall \sigma \in \mathcal{S}_i, \nexists \tau \in i^\mathcal{F}_\sigma ((\{\sigma\}\cup\tau)\subseteq\mathcal{R}_i)$

\vspace{3mm}

\item
There should not be any identifier for $\mathcal{S}_i$ according to the contrpositive of Identifiable Property's formula (5) and line (14) since it is prevented for any attribute in $\mathcal{S}_i$ to be in the same relation with an identifier.

\vspace{2mm}

$\mathcal{I}^\mathcal{F}_ {\mathcal{S}_i}=\emptyset$

\vspace{3mm}

\item
$\mathcal{A}_n$ cannot be an identifier according to line (15).

\vspace{2mm}

$\mathcal{A}_n \not\in \mathcal{I}^\mathcal{F}_ {\mathcal{S}_i}$

\vspace{3mm}

\item
Lines (13) and (16) are leading to a contradiction.

\vspace{2mm}

$\perp$

\vspace{3mm}

\item
First assumption of $\vee$ expression in line (7) in line (8) has been contradicted. Next, the second assumption is given below as $\mathcal{S}_i$ should be a part of an existing relation.

\vspace{2mm}

$ (\mathcal{R}_o \in \mathcal{R'} \wedge \mathcal{S}_i \subseteq \mathcal{R}_o)$

\vspace{3mm}

\item
Second $\wedge$ instantiation using line (18).

\vspace{2mm}

$\mathcal{S}_i \subseteq \mathcal{R}_o$

\vspace{3mm}

\item
First $\wedge$ instantiation using line (18).

\vspace{2mm}

$\mathcal{R}_o \in \mathcal{R'}$

\vspace{3mm}

\item
$\mathcal{S}_i$ cannot be part of any relation according to formula (12).

\vspace{2mm}

$\mathcal{S}_i \not\subseteq \mathcal{R}_o$

\vspace{3mm}

\item
Lines (19) and (21) are leading to a contradiction.

\vspace{2mm}

$\perp$

\vspace{3mm}

\item
Lines (17) and (22) are leading to a contradiction for line (5) which means that it is impossible to make an inference of association among the attributes in $\mathcal{S}_i$.

\vspace{2mm}

$\perp$

\vspace{3mm}

\item
End of proof by contradiction is reached, hence the theorem holds.

\vspace{2mm}

$\neg\mathcal{X}_\mathcal{L'}(\mathcal{S}_i)$

\vspace{3mm}

\end{enumerate}
\end{IEEEproof}

\end{document}